\documentclass[aps,twocolumn,showpacs,pra,citeautoscript,amsmath,amssymb,floatfix,superscriptaddress]{revtex4-1}

\usepackage{amsmath}
\usepackage{amssymb}
\usepackage{epsfig}
\usepackage{color,epstopdf}
\usepackage{amsthm, thm-restate}
\usepackage{hyperref}

\usepackage{mathtools}
\usepackage{graphicx}
\usepackage{hyperref}
\usepackage{paralist}
\usepackage{mathrsfs}
\usepackage[normalem]{ulem}
\usepackage{centernot}
\usepackage{stmaryrd}
\usepackage{epsfig}
\usepackage{float}
\usepackage{braket}

\usepackage{color}
\usepackage{amsthm}
\usepackage{xcolor}

\usepackage[makeroom]{cancel}
\usepackage{appendix}
\usepackage{cancel}

\usepackage{soul}
\usepackage[caption=false]{subfig}
\usepackage{algpseudocode}

\usepackage{thmtools} 
\usepackage{thm-restate}

\newtheorem*{theorem*}{Theorem}

\newtheorem{lemma}{Lemma}
\newtheorem*{lemma*}{Lemma}

\newtheorem{definition}{Definition}

\newtheorem{innercustomthm}{Theorem}
\newenvironment{customthm}[1]
{\renewcommand\theinnercustomthm{#1}\innercustomthm}
{\endinnercustomthm}

\newtheorem{innercustomdef}{Definition}
\newenvironment{customdef}[1]
{\renewcommand\theinnercustomdef{#1}\innercustomdef}
{\endinnercustomdef}

\DeclareMathOperator*{\Tr}{Tr}
\DeclareMathOperator*{\dist}{dist}

\DeclareRobustCommand{\openzero}{\leavevmode\hbox{0\kern-.55em0}}

\newcommand{\tr}[0]{\textnormal{Tr}}

\newcommand{\nay}[1]{{#1}}

\newcommand{\snote}[1]{{{ #1}}}

\newcommand{\prlsection}[1]{{\em {#1}.---~}}

\newcommand{\commentout}[1]{{{}}}

\begin{document}
	\title{Fundamental limitations to local energy extraction in quantum systems}
	\author{\'Alvaro M. Alhambra}
	\affiliation{Perimeter Institute for Theoretical Physics, Waterloo, ON N2L 2Y5, Canada}
	\email[e-mail address: ]{aalhambra@perimeterinstitute.ca}
	\author{Georgios Styliaris}
\affiliation{Department of Physics and Astronomy, and Center for Quantum Information
	Science and Technology, University of Southern California, Los Angeles,
	California 90089-0484}
	 \author{Nayeli A. Rodr\'{i}guez-Briones}
    \affiliation{Institute for Quantum Computing, University of Waterloo, Waterloo, Ontario, N2L 3G1, Canada}
   \affiliation{Department of Physics \& Astronomy, University of Waterloo, Waterloo, Ontario, N2L 3G1, Canada}
    \affiliation{Perimeter Institute for Theoretical Physics, Waterloo, ON N2L 2Y5, Canada}
	\author{Jamie Sikora}
	\affiliation{Perimeter Institute for Theoretical Physics, Waterloo, ON N2L 2Y5, Canada}
	\author{Eduardo Mart\'{i}n-Mart\'{i}nez}
\affiliation{Department of Applied Mathematics, University of Waterloo, Waterloo, Ontario, N2L 3G1, Canada}
\affiliation{Institute for Quantum Computing, University of Waterloo, Waterloo, Ontario, N2L 3G1, Canada}
\affiliation{Perimeter Institute for Theoretical Physics, Waterloo, ON N2L 2Y5, Canada}

	\begin{abstract}
	We examine when it is possible to locally extract energy from a bipartite quantum system in the presence of strong coupling and entanglement, a task which is expected to be restricted by entanglement in the low-energy eigenstates. 
	%
	We fully characterize this distinct notion of ``passivity'' by finding necessary and sufficient conditions for such extraction to be impossible, using techniques from semidefinite programming. This is the first time in which such techniques are used in the context of energy extraction, which opens a way of exploring further kinds of passivity in quantum thermodynamics.
    %
    %
    %
    We also significantly strengthen a previous result of Frey et al., 
    by showing a physically relevant quantitative bound on the threshold temperature at which this passivity appears.
    Furthermore, we show how this no-go result also holds for thermal states in the thermodynamic limit, provided that the spatial correlations decay sufficiently fast, and we give numerical examples.

	\end{abstract}
	
	\maketitle
	
\prlsection{Introduction} In the macroscopic regime, in which thermodynamic systems typically exchange energy via weak interactions, the possible flows of energy between them are easily understood in terms of the usual laws of thermodynamics. These laws, however, may become less relevant for systems where the fluctuations 
	and the particulars of the interaction between the micro-constituents are important. 
	Moreover, in
	the microscopic regime
	quantum effects due to e.g. coherence or entanglement appear, and a natural question arises: how do those effects alter the flows of energy in and out of the system?
	
	For the task of extracting energy locally from a bipartite 
	system, one could expect the following: if the low-energy eigenstates of the system display entanglement, there are limitations when trying to get closer to them only by means of local maps (since one cannot approach entangled states with local operations). While it could be possible to decrease the energy of the system up to some mixture of those low-energy eigenstates, trying to drive the system to a lower energy state can correspond to increasing the correlations in the system beyond what is possible via local operations alone.
	
Inspired by this intuition, here we focus on the problem of cooling interacting multipartite systems to which only local access to a single subsystem is granted.
We explore the most general type of local access to quantum systems, which is given by the CPTP maps \cite{nielsen2002quantum}, making our results relevant for any physical platform in which the subsystems are spatially separated.

	This problem was first studied in Frey et al.~\cite{frey2014strong}, who gave a set of sufficient conditions for the impossibility of energy-yielding via arbitrary local operations. They called this phenomenon strong local passivity (which we refer to here as CP-local passivity), and showed that having a non-degenerate ground state with full Schmidt rank is a \emph{sufficient} condition for the system to exhibit it, given a large enough population in the ground state. Here, we build on their results in two ways: $i)$ we find \emph{necessary and sufficient} conditions for this energy extraction to be impossible and $ii)$ we strengthen the set of physically motivated sufficient conditions found in \cite{frey2014strong}, by finding explicit bounds for the ground state population and critical temperature for which the system displays CP-local passivity. We also prove that these sufficient conditions hold for systems of arbitrary size provided that the spatial correlations are weak, thus extending the presence of CP-local passivity to strongly-coupled heat baths in the thermodynamic limit. Furthermore, we highlight the relevance of the necessary and sufficient conditions we find by constructing examples where none of the sufficient conditions are met. 
	
	%
	We also show that this effect of CP-local passivity, unlike the usual notion of passivity, should only be of fundamental relevance in quantum scenarios. In states without coherence or entanglement, it can only happen if the support of the states is fine-tuned and/or the Hamiltonian is sufficiently degenerate, which constitute very strong restrictions.
	
	  
	  
	 \prlsection{Setting} Let $\mathcal{H}_A \otimes \mathcal{H}_B$ be the Hilbert space associated with quantum systems $A$ and $B$, with global Hamiltonian $H_{AB}$. Given a state $\rho_{AB}$, the maximum extractable energy under a local map on $A$ is
	%
	%
	\begin{align}\label{eq:sdp1}
	\Delta E_{(A)B}&= \min_{\mathcal{E}_A} 	\Delta E_{(A)B}^{\mathcal{E}_A} \\&\coloneqq \min_{\mathcal{E}_A} \tr[H_{AB}(\mathcal{E}_A\otimes \mathcal{I}_B)\rho_{AB}]-\tr[H_{AB}\rho_{AB}],\nonumber
	\end{align}
	%
    %
	%
	%
	where $\mathcal{I}_B$ is the identity channel on $B$, and the optimization is over the whole set of CPTP maps on $A$. 
	The above optimization can be easily written as a   \emph{semidefinite program} (see \cite{boyd2004convex,watrous2018theory} for introductory references to the subject). 
	\snote{Therefore, it is very practical to calculate $\Delta E_{(A)B}$ and to find the CPTP map which minimizes the energy. Moreover, we see that}
	energy cannot be extracted when this quantity is zero, which motivates the following definition.

	
	\begin{definition}\label{def:SLP}
		[CP-local passivity] The pair $\{\rho_{AB},H_{AB}\}$ is CP-local passive with respect to subsystem $A$ if and only if
		\begin{equation}
		\Delta E_{(A)B}= 	\Delta E_{(A)B}^{\mathcal{I}_A} = 0.
		\end{equation}
	\end{definition}
	That is, a system is CP-local passive if the best local strategy for extracting energy (as measured by the global Hamiltonian $H_{AB}$) is to act trivially on it. The word \emph{passive} is used here in analogy to the commonly known passive states \cite{lenard1978thermodynamical}, from which energy cannot be extracted under unitary maps. 
Throughout, we assume that the time evolution given by the Hamiltonian $H_{AB}$ does not play a role. This means that this setting applies to situations in which the local actions happen quickly, in the same spirit as that of fast local quenches or pulses in other quantum thermodynamic settings \cite{gallego2014thermal,perarnau2018strong}.

Let us now outline how this might be possible. 
First, let us rewrite the term corresponding to the average energy of the system after applying a local map, as follows:
		\begin{equation} \label{CJeqn}
		\tr[H_{AB}(\mathcal{E}_A\otimes \mathcal{I}_B)\rho_{AB}] = \tr[C_{AA'} E_{A A^{'}}]. 
		\end{equation} 
where
$E_{AA'}$ is the Choi-Jamio\l kowski operator for an arbitrary channel ${\mathcal{E}_A : A \to A'}$, and $C_{AA'} \in \mathcal{H}_{A} \otimes \mathcal{H}_{A'}$ the Hermitian operator $C_{AA'}\equiv \tr_B[\rho_{AB}^{\Gamma_A}H_{A'B}]$, with $\rho_{AB}^{\Gamma_A}$ the partial transpose on $A$ \footnote{The Choi-Jamio\l kowski operator of a quantum channel $\mathcal{E}_{A\rightarrow A}(\cdot)$ is defined as $E=d_A \mathcal{E}_{A\rightarrow A} \otimes \mathcal{I} \ket{\Phi}\bra{\Phi} $, the result of applying it to an un-normalized maximally entangled state $d_A \ket{\Phi}\bra{\Phi}=\sum_{i,j} \ket{i_A}\bra{j_A}\otimes \ket{i_{A'}}\bra{j_{A'}}$ on the Hilbert space of $A$ and a copy $A'$. Note that the partial transpose is with respect to the same basis as the one chosen for the Choi-Jamio\l kowski operator.}.
		
Let us now assume that CP-local passivity holds, such that for all $E_{AA'}$ the energy of the system does not decrease after the local action:
	\begin{align}
		\!\!\!\!\! \tr[C_{AA'} E_{A A^{'}}]
		& \ge  \tr[H_{AB}\rho_{AB}],
		\end{align}
We can rewrite the right hand side, using the fact that $E_{AA'}$ satisfies $\tr_{A'}[E_{AA'}]=\mathbb{I}_{A}$, and defining  $d_A\ket{\Phi}\bra{\Phi}$ as the 
		Choi-Jamio\l kowski operator for the identity channel, as
		\begin{align}
		\tr[H_{AB}\rho_{AB}]&=\tr \left[d_A\ket{\Phi}\bra{\Phi} C_{AA'} \right] \\ &=\tr_A \left[ \tr_{A'}[ d_A\ket{\Phi}\bra{\Phi} C_{AA'}] \tr_{A'}[E_{AA'}]\right] \nonumber \\ &=
		\tr[(\tr_{A'}[d_A\ket{\Phi}\bra{\Phi} C_{AA'}]\otimes \mathbb{I}_{A'})E_{AA'}].
	\nonumber
		\end{align}
		
Since this holds for all $E_{AA'}$, this suggests that CP-local passivity will hold whenever the following operator inequality is true 
    	\begin{equation}\label{eq:condition}
		C_{AA'} \ge \tr_{A'}[d_A\ket{\Phi}\bra{\Phi} C_{AA'}]\otimes \mathbb{I}_{A'}.   
		\end{equation}

	
	
	\prlsection{Complete conditions} The previous inequality in fact gives the necessary and sufficient condition. This constitutes our first main result:
	%
    %
	\begin{restatable}{thm}{mainth}\label{th:main1}
		The pair $\{\rho_{AB},H_{AB} \}$ is CP-local passive with respect to subsystem $A$ if and only if $\tr_{A'}[d_A\ket{\Phi}\bra{\Phi} C_{AA'}]$ is Hermitian and 
		\begin{equation}\label{eq:condition}
		C_{AA'}-\tr_{A'}[d_A\ket{\Phi}\bra{\Phi} C_{AA'}]\otimes \mathbb{I}_{A'} \ge 0,    
		\end{equation}
		where \nay{ $\mathcal{H}_{A'}$ is a copy of the Hilbert space $\mathcal{H}_{A}$,
	} $C_{AA'} \in \mathcal{H}_{A} \otimes \mathcal{H}_{A'}$  is a Hermitian operator defined as $C_{AA'}\equiv \tr_B[\rho_{AB}^{\Gamma_A}H_{A'B}]$, with $\rho_{AB}^{\Gamma_A}$ the partial transpose on $A$, and $d_A \ket{\Phi}\bra{\Phi}$ the (maximally entangled) Choi-Jamio\l kowski operator of the identity channel. 
	\end{restatable}
    Notice that Eq.~\eqref{eq:condition} only depends on $\rho_{AB}$ and $H_{AB}$ through the operator $C_{AA'}$. In fact, Eq.~\eqref{CJeqn} guarantees that this operator contains all the information about how much energy can be extracted through local operations. 
	Once it is constructed, the operator inequality can be easily checked to find whether the pair $\{\rho_{AB},H_{AB}\}$ is CP-local passive or not. If it is not, the semidefinite program 
	can be solved to find the amount of energy that can be extracted, \snote{as well as the minimizing CPTP map}. The proof can be found in the Supplemental Material \cite{supp}, together with details on semidefinite programming duality theory, which we use in a similar manner as in the proof of the Holevo-Yuen-Kennedy-Lax conditions for quantum state discrimination \cite{holevo1973statistical,holevo1973statistical2,yuen1970optimal,yuen1975optimum}.

	\nay{On top of this characterization, we show that the condition of Theorem \ref{th:main1} is robust to errors, by using a recent result concerning convex channel optimization problems \cite{coutts2018certifying}. }
	Roughly, if the operator on the LHS of Eq.~\eqref{eq:condition} has smallest eigenvalue $-\varepsilon \le 0$, then the amount of energy that can be extracted is bounded as $	\Delta E_{(A)B} \ge -\varepsilon \, d_A$. We give the precise statement and the proof in the Supplemental Material \cite{supp}. 
	
	\prlsection{Sufficient conditions} 
	The condition of Theorem \ref{th:main1}, even though it is simple to verify, makes no direct reference to physical properties of the pair $\{\rho_{AB},H_{AB}\}$. It is important, however, to find physically relevant situations in which CP-local passivity holds. To that end, we derive sufficient conditions for steady states $\rho_{AB} = \sum_{i = 0} ^ {d_A \times d_B -1} p_i \ket{E_i} \bra{E_i}$ of Hamiltonians $H_{AB} = \sum_{i=0}^{d_A \times d_B -1} E_i \ket{E_i}\bra{E_i}$ of full Schmidt rank with a non-degenerate ground state. Steady states are always trivially CP-local passive for $p_0 = 1$, and Frey et al. \cite{frey2014strong} found qualitative conditions under which there exists a threshold ground state population $p_*$ such that the pair $\{\rho_{AB},H_{AB}\}$ remains CP-local passive for all $p_0 \ge p_*$ . Here, we provide explicit upper bounds on $p_*$ in terms of ground state entanglement and the energy gap with the first excited state.
	
	\begin{restatable}[Threshold ground state population]{thm}{suff} \label{th:suff1}
		Let the ground state $\ket{E_0}$ of the Hamiltonian $H_{AB}$ be non-degenerate and with full Schmidt rank. All pairs $\{\rho_{AB},H_{AB}\}$ with $\rho_{AB} = \sum_{i} p_i \ket{E_i} \bra{E_i}$ and $p_0 \ge p_*$ are CP-local passive with respect to A, with the threshold ground state population bounded from above by 
		\begin{align}\label{eq:pbound} 
		p_*  \le  \left( 1 + \frac{E_1   (q^{AB}_{0,\min})^2}{\max_{i\ge 1} \left[ E_i  (q^{AB}_{i,\max})^2 \right]} \right)^{-1} . 
		\end{align}
		$\{  q^{AB}_{i,\alpha} \}_{\alpha=0}^{d_A-1} $ denotes the Schmidt coefficients of $\ket{E_i}$ and $q^{AB}_{i,\min} \equiv \min_{\alpha} \left[ q^{AB}_{i,\alpha} \right]$, $q^{AB}_{i,\max} \equiv \max_{\alpha} \left[ q^{AB}_{i,\alpha} \right]$.
	\end{restatable}
	See the Supplemental Material \cite{supp} for the proof, and an example illustrating the tightness of the bound. The idea behind it is that, if the ground state population is high enough, the energetic changes caused by any CPTP map will be dominated by the energy gained by exciting the ground state into higher energy levels, making the total change non-negative. 
	
For thermal states, this result implies that, if the ground state has full Schmidt rank, there exists a threshold temperature $T_*>0$ below which CP-local passivity holds (note that if $T=0$, CP-local passivity holds trivially). Moreover, this threshold temperature is such that 
	\begin{align} 
	\braket{H}_{\beta_*}\ge E_1 p_0 (q^{AB}_{0,\min})^2	\,\;,
	\label{eq:Tb}
	\end{align}
where $\braket{H}_{\beta_*}$ is the average energy in the thermal state of inverse temperature $\beta_*$.
 
We now describe when we expect this bound to be of importance. 
An entangled state of full Schmidt rank is typical in first-neighbor interactions where the local Hamiltonians do not commute with the interaction ones.
However, given that $q^{AB}_{i,\min} \le 1/d_A$, the bound weakens as the size of $A$ grows (and it trivializes once $d_A>d_B$). Also, a unique ground state and a finite energy gap is needed. On top of that, frustration is required, as we show in the following. Let us rewrite the Hamiltonian as $H_{AB}=H_A+H_B+V_{AB}$. The \emph{frustration energy} of $H_{AB}$ is defined as
	\begin{equation}
	E_f \equiv E_0^{H_{AB}}-E_0^{H_{A}+H_B}-E_0^{V_{AB}},
	\end{equation}
	where $E_0^{H}$ is the ground state energy of Hamiltonian $H$. This quantity measures the degree of frustration of $H_{AB}$ w.r.t. a particular decomposition into local and interaction terms of $H_{AB}$.
	The main result of \cite{dawson2004frustration} then states that
	\begin{equation}
	\frac{E_f}{\max_{i \in \{A,B\}} E_1^i} \ge 1- q^{AB}_{0,\max} \ge (d_A-1)q^{AB}_{0,\min},
	\end{equation}
	where $E_1^i$ is the gap of the local Hamiltonian $H_i$.
	This shows precisely that a certain level of frustration is necessary to have entanglement (in particular, with full Schmidt rank) in a unique ground state.
	
Note however, that while these conditions are sufficient, they are by no means necessary. In fact, we provide simple examples of pairs that are CP-local passive but in which $i)$ the ground state is not entangled, $ii)$ the ground state is degenerate and $iii)$ the state is not diagonal in the energy eigenbasis. These can be found in the Supplemental Material \cite{supp}.

	\prlsection{Thermodynamic limit} The bound in Eq.~\eqref{eq:pbound} trivializes when the system $B$ becomes very large, as the energy $E_i$ grows with it. However, we show that for thermal states with weak spatial correlations, one can increase the size of system $B$ indefinitely without breaking CP-local passivity. Hence, this phenomenon can hold even in the thermodynamic limit. First, we need the following definition. 
	
	%
	\begin{definition}[Clustering of correlations]
		A state $\rho$ on a finite square lattice $\mathbb Z^D$ has $\epsilon(l)$-clustering of correlations if
		\begin{equation}\nonumber
		\max_{M,N} \left| \tr[M\otimes N \rho]-\tr[M\rho]\tr[N\rho] \right| \le \vert\vert M \vert \vert \, \vert\vert N \vert\vert \, \epsilon(l),
		\end{equation}
		where the operator $M$ has support on region $A$ and $N$ on region $B$, and $l\le \text{dist} (A,B)$, with $\text{dist} (A,B)$ the Euclidean distance on the lattice. 
	\end{definition}

	For a state $\rho$ with $\epsilon(l)$-clustering of correlations, it is reasonable to expect that CP-local passivity is only determined by the vicinity of the region in which we act. We make this intuition precise in the following result. Let $H_{AB}$ be a Hamiltonian on regions $A,B$ in a $d$-dimensional finite square lattice.  Let $B_1, B_2$ be any splitting of $B$ (see Fig.~\ref{fig:regions}), with $l \equiv \text{dist} (A,B_2)$ the distance over which $B_1$ shields $A$ from $B_2$, with a boundary between $B_1,B_2$ of size $\vert \partial B_2 \vert$. More precisely, $H_{AB}$ takes the form
	\begin{equation}\label{eq:hamAB12}
	H_{AB}=H_A+V_{AB_1}+H_{B_1}+V_{B_1B_2}+H_{B_2}.
	\end{equation}
	We shall denote $H_{AB_1} \equiv H_A+V_{AB_1}+H_{B_1}$, and define $E_i^{AB_1},q_{i,\alpha}^{AB_1}$ as the eigenvalues and Schmidt coefficients of $H_{AB_1}$. Let region $S \subseteq A$ be such that no site in $S$ interacts with any site outside of the region $A$ under $H_{AB}$ (see Fig.~\ref{fig:regions}). The result is as follows:
	\begin{restatable}{thm}{largeth} \label{th:suffcorr}
		Consider a Hamiltonian $H_{AB}$ as in Eq.~\eqref{eq:hamAB12} and let $\tau_{AB}^{\beta}=e^{-\beta H_{AB}}/Z_{AB}$ be its thermal state with $\epsilon(l)$-clustering of correlations. There exists a finite temperature $\beta_*$ such that all pairs $\{\tau_{AB}^{\beta},H_{AB}\}$ with $\beta \ge \beta_*$ are CP-local passive with respect to local operations on $S$ if the regions $B_1,B_2$ can be chosen such that
		\begin{align}\label{eq:suffineq}
		E_1^{AB_1} \left( q_{0,\min} ^{AB_1}\right)^2 > \lambda(l) \,\;,
		\end{align}
		where 
		\begin{align} 
		\lambda(l) = K d_A^2 \,\vert\vert H_A \vert\vert \, |\partial B_2| \, (\epsilon(l/2)+c_1 e^{-c_2 l}) \,\;. \label{lambda_expr}
		\end{align}
		Moreover, $\beta_*$ is such that 
		\begin{align}\label{eq:boundbeta}
		\tr[e^{-\beta_* H_{AB_1}}]^{-1} &\le \left(1+\frac{\lambda(l)}{\max_{i\ge 1} \left[ E^{AB_1}_i (q^{AB_1}_{i,\max})^2\right ] } \right) \nonumber \\ & \times  \left( 1 + \frac{E^{AB_1}_1  (q^{AB_1}_{0,\min})^2} {\max_{i\ge 1} \left[ E^{AB_1}_i (q^{AB_1}_{i,\max})^2 \right] }\right)^{-1}.
		\end{align}

		where $K,c_1,c_2 > 0$ are constants.
	\end{restatable}
	\begin{figure}[t]
		\includegraphics[width=0.7\columnwidth]{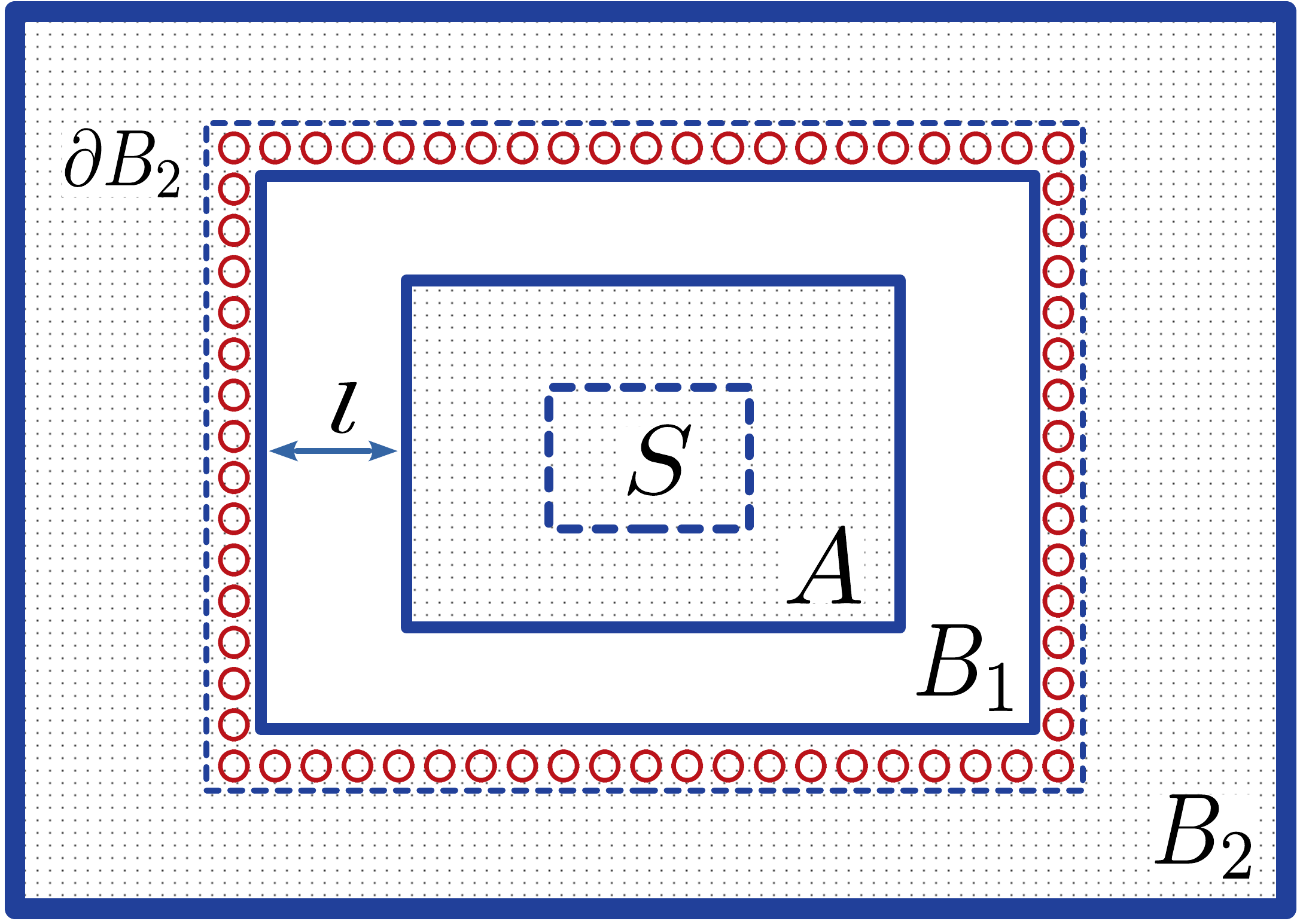}
		\caption{Regions on the lattice for Theorem \ref{th:suffcorr}. The map acts on a region $S \subset A$, which is shielded from the region $B_2$ by $B_1$, by a distance of $l$. The boundary of the lattice between $B_1$ and $B_2$ is defined as $\partial B_2$ and has a number of sites $\vert \partial B_2 \vert$.}
		\label{fig:regions}
	\end{figure}
		  
	The proof can be found in the Supplemental Material \cite{supp}. It relies on a result from \cite{brandao2016finite} (which builds on \cite{hastings2007quantum}),
	that shows how clustering of correlations implies that the marginals of many-body thermal states can be efficiently estimated by looking only at subregions of the lattice.
	%
	%
	Crucially, the bound on $\beta_*$ in Eq. \eqref{eq:boundbeta} only depends on parameters of the Hamiltonian $H_{AB_1}$ and on $\lambda(l)$, and is independent of $B_2$ (in particular, on its size) except for the boundary factor $|\partial B_2| \sim l^{D-1}$, with $D$ the dimension of the lattice. Hence, the best possible bound on $\beta_*$ for an arbitrary system size is achieved by choosing a partition $AB_1B_2$ such that the marginals on A of $\tau_{AB}$ and $\tau_{AB_1}$ are close enough, and the size of $AB_1$ is not too large to render the bound useless.


A choice of regions (or rather, the choice of $l$) giving a non-trivial bound is possible provided that the correlations of the thermal state decay fast enough. More concretely, as long as we can find an $l$ such that Eq. \eqref{eq:suffineq} holds, the upper bound on $p^*$ of Eq. \eqref{eq:boundbeta} is non-trivial. We expect this to be possible in a large class of models, as the gap rarely closes faster than polynomially with system size (if at all), and having an exponentially-decaying $\epsilon(l)$-clustering of correlations at finite temperature is a property of many lattice models \cite{araki1969gibbs,hastings2004decay,kliesch2014locality}.
%
%
In Fig.~\ref{fig:largeN}, we provide a numerical example of a model in which we calculate how the threshold temperature changes as we increase the system size. Note that the curves converge as $N$ becomes large, showing that larger system sizes do not affect the threshold temperature appreciably. 
	\begin{figure}[h]	
		\includegraphics[width=\columnwidth]{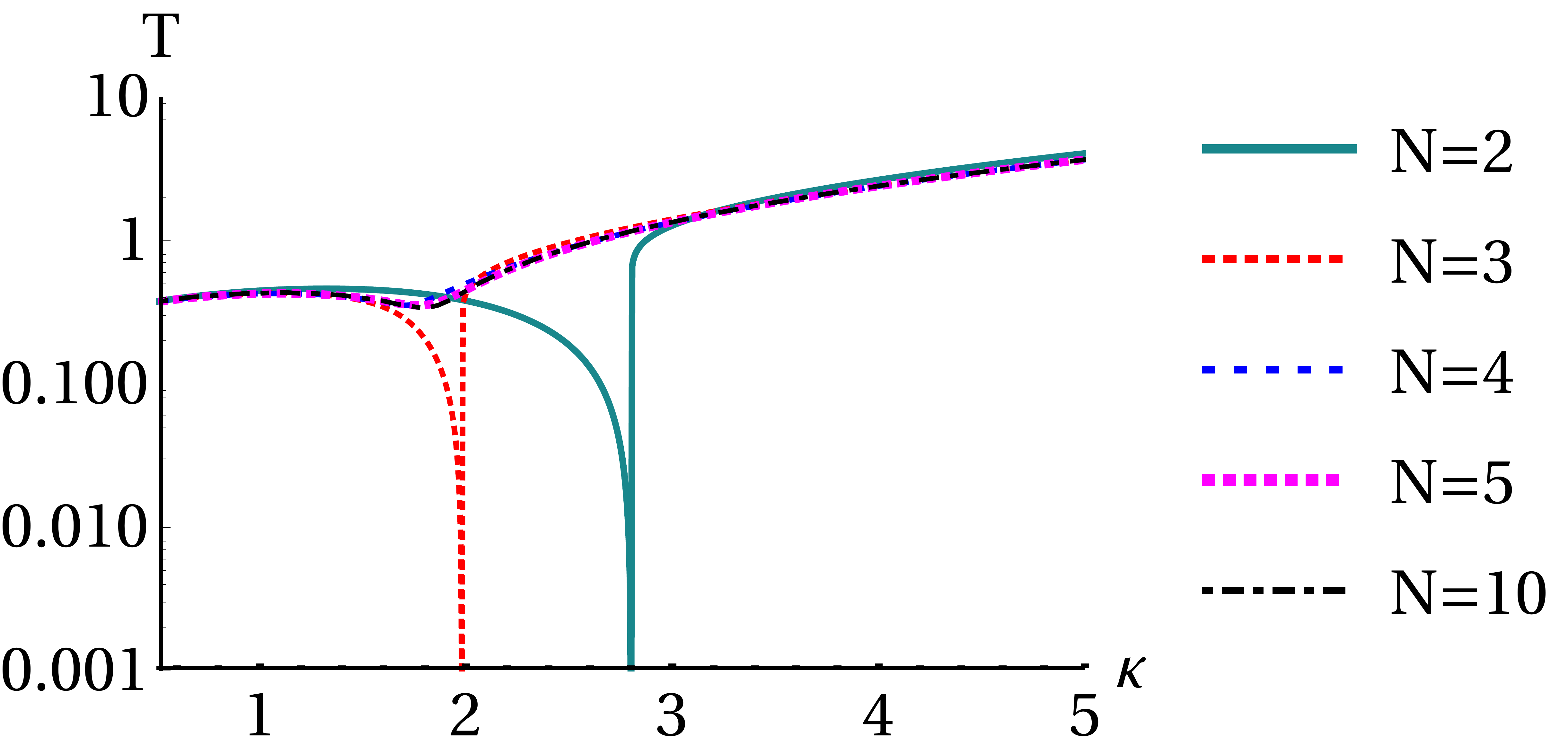}
		\caption{ Threshold temperature \nay{for} the 1D Hamiltonian $ H_{AB}=\sum^N_{l=1} \sigma^{(l)}_Z + \kappa(\frac{1+\gamma}{2}\sigma^{(l)}_X \sigma^{(l+1)}_X+\frac{1-\gamma}{2}\sigma^{(l)}_Y \sigma^{(l+1)}_Y)$ as a function of the coupling strength $\kappa$,  fixing the anisotropy parameter $\gamma=0.7$. The system $A$ on which the maps act is the leftmost qubit $l=1$. \nay{For $N>3$} the curves overlap, showing that increasing system $B$ beyond a certain size does not affect the threshold temperature appreciably. The threshold temperature was determined using the condition of Theorem \ref{th:main1}. }
		\label{fig:largeN}
	\end{figure}
	
		\prlsection{Classical CP-local passivity}
	This phenomenon can appear in certain classical situations (for instance, when the Hamiltonian is non-interacting and the initial state is $\rho_{AB}=\ket{0}\bra{0}\otimes \rho_B$), but we argue that either coherence or entanglement are necessary for it to be non-trivial. 
	\nay{We do this by showing that CP-local passivity, in a classical setting, only happens in very restricted situations.}
	Let us consider a fully classical model, with an incoherent state $\rho_{AB}$ and a Hamiltonian with product eigenstates, such that %
	\begin{align}
	H_{AB}=\sum_{i,j} E_{i,j} \ket{i}\bra{i} \otimes \ket{j} \bra{j},
	\\
	\rho_{AB}=\sum_{i,j} p_{i,j} \ket{i}\bra{i} \otimes \ket{j} \bra{j}.
	\end{align}
	Without loss of generality, we can order the energies such that $E_{i,j}\le E_{i+1,j}$ and $E_{i,j}\le E_{i,j+1}$.
	The optimal local cooling strategy is straightforward: map the initial eigenstates to the eigenstates of lower energy that can be accessed with local maps. Let us write
	\begin{align}
	\Delta E_{(A)B}^{\mathcal{E}_A} &= \tr[H_{AB}(\mathcal{E}_A\otimes \mathcal{I}_B)\rho_{AB}]-\tr[H_{AB}\rho_{AB}]\\ &=\sum_{i,k} \sum_j E_{ij} \sum_{l}p_{k l}\delta_{j,l} \bra{i}\mathcal{E}_A (\ket{k}\bra{k})\ket{i}-\delta_{i,k} \nonumber
	\\&\equiv \sum_{i,k} \tilde E_{i,k} \left( \bra{i}\mathcal{E}_A (\ket{k}\bra{k})\ket{i}-\delta_{i,k}\right),
	\end{align}
	where  $\tilde E_{i,k}=\sum_j E_{i,j}p_{k,j}$. 
	The optimal CPTP map is such that $\mathcal{E}^{\text{opt}}_A (\ket{k}\bra{k})=\ket{i_{k}^*}\bra{i_{k}^*}$  $\forall k$, 
	where $i^*_{k}=\text{argmin}_i \tilde E_{i,k}$, and thus
	\begin{align}\label{eq:classopt}
	\Delta E_{(A)B}= \Delta E_{(A)B}^{\mathcal{E}^{\text{opt}}_A} &= \sum_{k} \tilde E_{i^*_{k},\, k}-  \tilde E_{k, k},
	\end{align}
	which is non-negative if and only if $i^*_{k} = k \,\, \forall k$, in which case $\{\rho_{AB},H_{AB}\}$ is CP-local passive. This happens only if the matrix $\tilde E_{i,k}$ is such that the smallest number in each row (indexed by $k$) is in the diagonal.  This condition, however, can only be met by states with a particular support or by highly degenerate Hamiltonians. 
	To be more precise, let us look at the individual terms of Eq. \eqref{eq:classopt} for every $k>1$, 
	\begin{equation}\label{eq:iminusone}
	\tilde E_{k-1,k}-\tilde E_{k,k}=\sum_{j}(E_{k-1,j}-E_{k,j})p_{k,j}.
	\end{equation}
	Since $E_{k-1,j}-E_{k,j} \le 0$ by definition, the only way Eq. \eqref{eq:iminusone} can be non-negative is if either $p_{k,j}=0$ or $E_{k-1,j}=E_{k,j}$ $\forall j$, which constitutes a strong restriction on the support of the initial state and $H_{AB}$. For instance, no thermal state (with full support) of a Hamiltonian with any non-degeneracy on index $k$ will obey this condition.

	\prlsection{Discussion}
	We have found necessary and sufficient conditions for CP-local passivity, which take the form of a simple inequality of operators of size $d_A \times d_A$. We also derived simpler sufficient conditions that show definite physical situations in which this phenomenon appears, and we provide numerical examples illustrating the general picture. 
	
	Our proof of the necessary and sufficient conditions, of Theorem \ref{th:main1}, uses tools from the theory of convex optimization, widely used in quantum information, but which, apart from a few exceptions \cite{faist2015minimal,faist2018fundamental}, have not yet been exploited in quantum-thermodynamic contexts. In fact this is, to our knowledge, the first time that the theory of semi-definite-programming has been used in the context of energy extraction and passivity. We expect these tools to be of further use in similar situations in which the actions allowed on the state are limited in different physically motivated ways. The fact that we optimize over a linear function of the channels (the energy of the output) made the derivations particularly simple, but in fact recent results \cite{coutts2018certifying} easily allow for extensions to arbitrary non-linear functions.

A further set of previous results (e.g. \cite{oppenheim2002thermodynamical,jennings2010entanglement,perarnau2015extractable}) identify entanglement in the initial state as a useful resource in  energy extraction when one has access to global operations and the Hamiltonians are non-interacting. Here we explore a different side of the general picture, by showing that entanglement in the eigenstates can forbid the possibility of energy extraction via local operations when the interactions are strong.

	The underlying principle here is that entanglement in the low-energy eigenstates causes a fundamental lack of local control in systems at low  temperature, provided that the CPTP maps are fast compared to the dynamics of the system. 
	This effect could potentially also include
	 quenches and/or pulses that are commonly taken as the steps of quantum thermal cycles in which ``work" is exchanged\cite{gallego2014thermal,gelbwaser2015strongly,newman2017performance,perarnau2018strong}, in which case our results should put constraints on their regime in which those machines can perform. 
	%
	
	
A further study on CP-local passivity could be the characterization of scenarios in which this passivity can be circumvented by allowing classical communication.
This type of setting goes under the name of quantum energy teleportation (QET) \cite{hotta2008quantum,frey2013quantum,hotta2010energy}. 
Our necessary and sufficient conditions could help design better QET-based protocols, which have been applied both in quantum field theory \cite{hotta2010controlled} and algorithmic cooling in quantum information processing \cite{rodriguez2017correlation}.

	\quad \\
	
	\begin{acknowledgments} 
	\quad \\  
	\textit{Acknowledgments.}  
	The authors acknowledge useful discussions with Masahiro Hotta, Philippe Faist, Raam Uzdin, Marti Perarnau-Llobet and Mark Girard. This research was supported in part by Perimeter Institute for Theoretical Physics. Research  at  Perimeter  Institute  is  supported  by  the
	Government of Canada through the Department of Innovation, Science and Economic Development and by the Province of Ontario through the Ministry of Research, Innovation and Science. 
	N.R-B acknowledges support of CONACYT and Mike and Ophelia Lazaridis Scholarship. E. M-M. acknowledges support of the NSERC Discovery program as well as his Ontario Early Researcher Award.
	\end{acknowledgments}

\widetext	
		\appendix
	
	\section{Semidefinite Programming and a Proof of Theorem 1}\label{app:proof1}
	
	
	We start by giving a brief introduction to semidefinite programming and its duality theory. 
	A semidefinite program (SDP) is an optimization problem of the form 
	\begin{equation} \label{primal}
	\alpha = \inf \, \{ \tr(CX) : \Phi(X) = B, X \geq 0 \} 
	\end{equation} 
	where $X$ is the variable, $C$ and $B$ are Hermitian matrices, and $\Phi$ is a linear, Hermiticity-preserving superoperator. 
	
	Associated with every SDP is its dual, which is also an SDP, defined as 
	\begin{equation} \label{dual}
	\beta = \sup \, \{ \tr(BY) : \Phi^*(Y) \leq C \} 
	\end{equation} 
	where $Y$ is the dual variable and $\Phi^*$ is the adjoint of $\Phi$. 
	
	Note that if $X$ is a feasible solution (satisfies $\Phi(X) = B$, $X \geq 0$) and $Y$ is a dual feasible solution (satisfies $\Phi^*(Y) \leq C$) then we have 
	\begin{equation} \label{WD}
	\tr(CX) - \tr(BY) 
	=     
	\tr(CX) - \tr(\Phi(X)Y) 
	= 
	\tr(X(C - \Phi^*(Y))) 
	\geq 0
	\end{equation}
	since both $X$ and $C - \Phi^*(Y)$ are positive semidefinite. 
	This is known as Weak Duality, which states that $\alpha \geq \beta$. 
	Suppose we have a fixed feasible solution $X'$. Then if there exists a dual feasible solution $Y'$ satisfying ${\tr(X'(C - \Phi^*(Y'))) = 0}$, or equivalently 
	\begin{equation} \label{CS}
	X'C =  X \Phi^*(Y'),  
	\end{equation} 
	then this would certify that $X'$ is an optimal solution, via \eqref{WD}. 
	The condition \eqref{CS} is called \emph{complementary slackness}.  
	
	Under mild conditions, if $X'$ is an optimal solution to \eqref{primal}, then one can guarantee the existence of such a $Y'$ in the discussion above. 
	
	\begin{lemma} \label{SDPlemma}
		Suppose there exists $X > 0$ satisfying $\Phi(X) = B$. Then a feasible solution $X'$ is an optimal solution to $\eqref{primal}$ if and only if there is a dual feasible $Y'$ satisfying 
		\eqref{CS}. 
	\end{lemma} 
	The proof of this is beyond the scope of this discussion, and we refer the interested reader to the book \cite{boyd2004convex}.  
	Recall the SDP which solves for the optimal local channel in our problem, reproduced below 
	\begin{equation}  \label{eq26}
	\inf \, \{ \tr[C_{AA'} E_{A A^{'}}] : \tr_{A'}[E_{AA'}]=\mathbb{I}_{A}, \, E_{AA'} \geq 0 \}. 
	\end{equation} 
	Note that there exists $E_{AA'} > 0$ satisfying $\tr_{A'}[E_{AA'}]=\mathbb{I}_{A}$ (take a scalar multiple of $\mathbb{I}_{AA'}$ for example). 
	Thus, the conditions for Lemma~\ref{SDPlemma} are satisfied. 
	
	Using the fact that $\tr(\tr_{A'}[E_{AA'}] Y) = \tr(E_{AA'} (Y \otimes \mathbb{I}_{A'}))$, we can use the formula \eqref{dual} to write the dual of \eqref{eq26} as 
	\begin{equation} \label{eq:dual2}
	\sup \, \{ \tr(Y) : Y \otimes \mathbb{I}_{A'} \leq C_{AA'} \}. 
	\end{equation} 
	Using Lemma~\ref{SDPlemma}, we have that $E_{AA'} = d_A \ket{\Phi}\bra{\Phi}$ (i.e. the identity channel) is an optimal channel if and only if there exists a dual feasible $Y'$ satisfying \eqref{CS}, which in this case can be written as 
	\begin{equation} 
	d_A \ket{\Phi}\bra{\Phi} C_{AA'} = d_A \ket{\Phi}\bra{\Phi} (Y' \otimes \mathbb{I}_{A'}). 
	\end{equation} 
	By taking the partial trace of both sides, we have that 
	\begin{equation} 
	\tr_{A'}[d_A \ket{\Phi}\bra{\Phi} C_{AA'}] = \tr_{A'}[d_A \ket{\Phi}\bra{\Phi} (Y' \otimes \mathbb{I}_{A'})] = 
	Y'   
	\end{equation} 
	noting again that $\tr_{A'}[d_A \ket{\Phi}\bra{\Phi}] = \mathbb{I}_{A}$. 
	Since this $Y'$ is dual feasible, we know from Eq.~\eqref{eq:dual2} that 
	\begin{equation} \label{necsuffcond}
	\tr_{A'}[d_A \ket{\Phi}\bra{\Phi} C_{AA'}] \otimes \mathbb{I}_{A'} \leq C_{AA'}, 
	\end{equation} 
	and $\tr_{A'}[d_A \ket{\Phi}\bra{\Phi} C_{AA'}]$ is Hermitian. 
	Thus, if the identity channel is optimal, i.e., no energy can be extracted, then Eq.~\eqref{necsuffcond} holds and $\tr_{A'}[d_A \ket{\Phi}\bra{\Phi} C_{AA'}]$ is Hermitian.
	
	Conversely, recall from Eq.~(3) in the manuscript that 
	\begin{equation} 
	\tr[H_{AB}(\mathcal{E}_A\otimes \mathcal{I}_B)\rho_{AB}] = \tr[C_{AA'} E_{A A^{'}}]
	\end{equation} 
	where $E_{A A^{'}}$ is the Choi-Jamio\l kowski operator of the CPTP map $\mathcal{E}_A$. 
	Thus, if Eq.~\eqref{necsuffcond} is true, and $\tr_{A'}[d_A \ket{\Phi}\bra{\Phi} C_{AA'}]$ is Hermitian, then 
	\begin{equation} 
	\tr[H_{AB}(\mathcal{E}_A\otimes \mathcal{I}_B)\rho_{AB}] = \tr[C_{AA'} E_{A A^{'}}]
	\geq \tr[[\tr_{A'}[d_A \ket{\Phi}\bra{\Phi} C_{AA'}] \otimes \mathbb{I}_{A'}] E_{A A^{'}}]. 
	\end{equation} 
	Eq.~(5) in the manuscript shows that 
	\begin{equation} 
	\tr[[\tr_{A'}[d_A \ket{\Phi}\bra{\Phi} C_{AA'}] \otimes \mathbb{I}_{A'}] E_{A A^{'}}] = \tr[\rho_{AB} H_{AB}]. 
	\end{equation} 
	Thus, the action of any local CPTP map $\mathcal{E}_A$  will not decrease the energy, as desired.

	This proof is a reformulation of that of the necessary and sufficient conditions for the problem of quantum state discrimination due to Holevo-Yuen-Kennedy-Lax ~\cite{holevo1973statistical,holevo1973statistical2,yuen1970optimal,yuen1975optimum}. 
	This problem is beyond the scope of this work, but it can be cast as an optimization over quantum channels as in Eq.~\eqref{eq26} for a different $C_{AA'}$ matrix. See Ref.~\cite{coutts2018certifying} for more details and for generalizations of this proof.

	\subsection{Robust version of Theorem 1} \label{app:epsilon}
	
	We also discuss here the possibility of the identity channel being \emph{almost optimal}. 
	By applying a result in \cite{coutts2018certifying} to our problem, we have that 
	\begin{equation} 
	\tr(d_A\ket{\Phi}\bra{\Phi} C_{AA'}) - \min_{E_{AA'}} \tr(E_{AA'} C_{AA'}) 
	\leq \varepsilon \, d_A 
	\end{equation} 
	where 
	\begin{equation} \label{eps_defn}
	\varepsilon = \left| \lambda_{\min}
	\left( 
	C_{AA'} - \frac{1}{2} \left(
	\tr_{A'}[d_A\ket{\Phi}\bra{\Phi} C_{AA'}] \otimes \mathbb{I}_{A'} 
	+ h.c. \right) 
	\right) \right|
	\end{equation} where we see that $\varepsilon$ is in fact a measure of how far away the matrix is from being Hermitian and positive.
	
	The proof of this is rather simple in this case. 
	Let $\gamma = \tr(d_A\ket{\Phi}\bra{\Phi} C_{AA'}) \in \mathbb{R}$ for convenience.   
	Define 
	\begin{equation} 
	Y = \frac{1}{2} \left( \tr_{A'}[d_A\ket{\Phi}\bra{\Phi} C_{AA'}] + h.c. \right) - \varepsilon \, \mathbb{I}_{A}. 
	\end{equation}   
	Then we have 
	\begin{equation} 
	C_{AA'} - Y \otimes \mathbb{I}_{A'} 
	\geq 0 
	\end{equation} 
	using Eq.~\eqref{eps_defn}. 
	Therefore, $Y$ is dual feasible and has value 
	$\tr(Y) = \gamma - \varepsilon \, d_A$. 
	By weak duality, we have that 
	\begin{equation} 
	\min_{E_{AA'}} \tr(E_{AA'} C_{AA'}) \geq \tr(Y) = \gamma - \varepsilon \, d_A. 
	\end{equation} 
	Rearranging the above inequality gives us the result.

	\section{Sufficient physical conditions are not necessary (numerical examples)}\label{app:num}
	
	Here we show that
	the
	sufficient physical conditions for CP-local passivity, presented first in \cite{frey2014strong} and strengthened in the present work, are by no means necessary. 
	More concretely, we find situations in which either an entangled or non-degenerate ground state are not present. Furthermore, we relax the assumption of $\left[\rho_{AB},H_{AB}\right]=0$, finding that this is not necessary for CP-local passivity.
	
	\subsection{ CP-local passivity without entangled ground state}
	
	The system consists of a pair of qubits A and B, with Hamiltonian
	\begin{equation}
	H_{AB}= \frac{1}{2}\omega\sigma_z^A+\frac{1}{2}\omega\sigma_z^B+\frac{\kappa}{2}(\sigma_x^A\sigma_x^B+\sigma_y^A\sigma_y^B), 
	\label{eq:hamab}
	\end{equation}
	where $\kappa>0$ is the coupling strength. 
	When fixing $\omega=-2$, 
	the eigenstates of the system are given by
	\begin{equation}
	\displaystyle
	\{\ket{00},\quad\ket{11},\quad\frac{1}{\sqrt{2}}\left(\ket{10}-\ket{01}\right),\quad\frac{1}{\sqrt{2}}\left(\ket{01}+\ket{10}\right)\}    
	\end{equation}
	with corresponding eigenenergies  $\{-2,2,-\kappa,\kappa\}$, respectively. Note that for $\kappa<2$ the ground state is non-degenerate but separable $\ket{E_0}=\ket{00}$; and for $\kappa=2$, the ground state is degenerate.
	
	We find the threshold temperature in the region of ground state degeneracy,
	by using the necessary and sufficient conditions presented in our theorem 1. In Fig.\ref{fig:non-entangled-ground-state}, we show the results for $T^*$ as a function of the strength coupling, where we find that the system can be CP-local passive even without an entangled ground state.
	

	\begin{figure}[h]
		\includegraphics[width=0.6\textwidth]{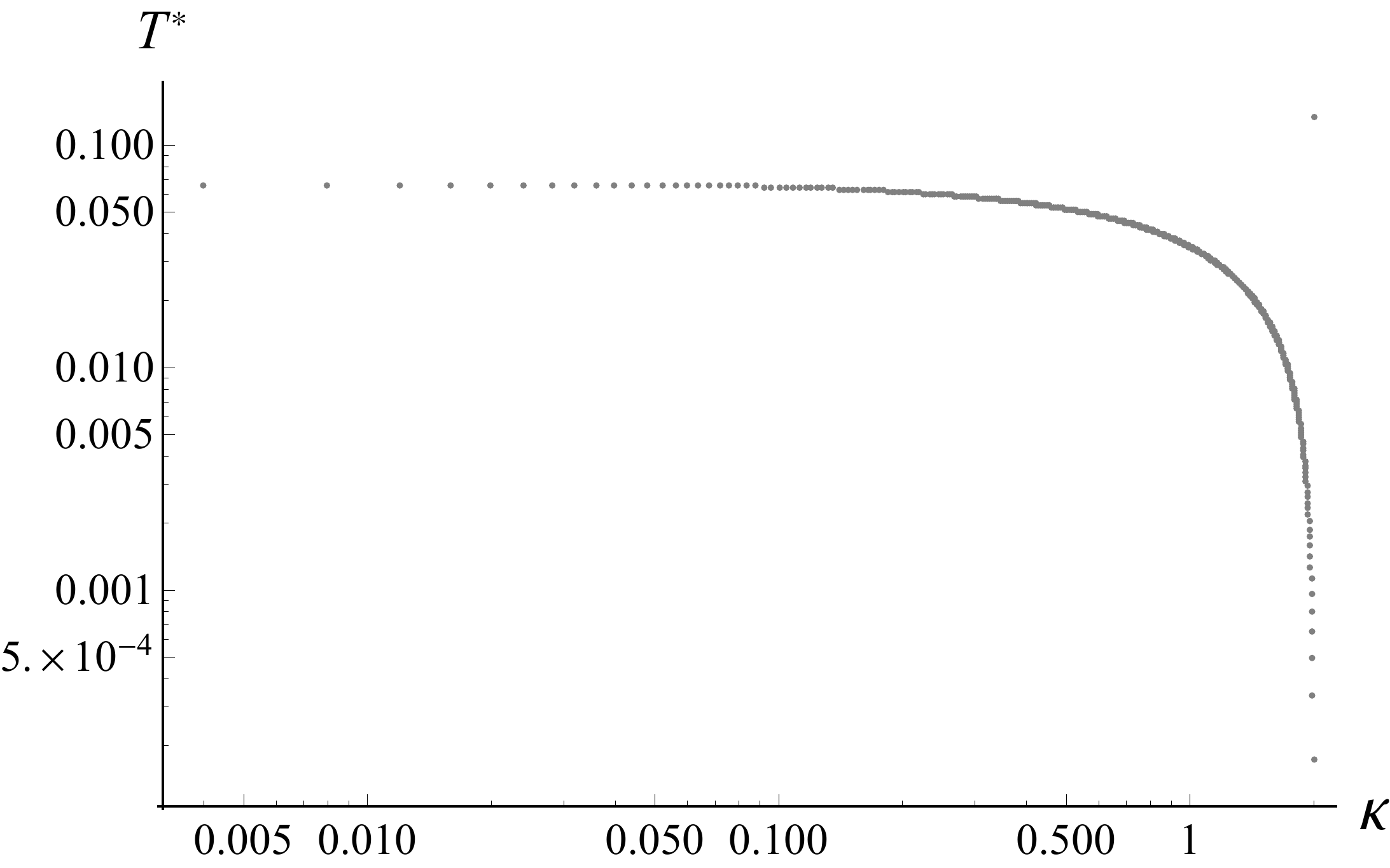}
		\caption{\textbf{CP-local passivity without entangled ground state: } Here we show the threshold temperature $T^*$ for a pair of qubits with Hamiltonian $H_{AB}= \frac{1}{2}\omega\sigma_z^A+\frac{1}{2}\omega\sigma_z^B+\frac{\kappa}{2}(\sigma_x^A\sigma_x^B+\sigma_y^A\sigma_y^B)$ as a function of the coupling strength $\kappa$, with fix $\omega=-2$. Even though the ground state is separable in this region of strength coupling  ($\ket{E_0}=\ket{00}$ for $-2<\kappa<2$), it is still possible to obtain CP-local passivity.  The dip for $k\to2$, approached from the left, occurs as the ground state gets close to be degenerate.}
		\label{fig:non-entangled-ground-state}
	\end{figure}

	\subsection{ CP-local passivity with degenerate ground state}
	
	Consider a pair of qubits A and B, with Hamiltonian $H_{AB}=\kappa\sigma_x^A\sigma_x^B$, where $\kappa>0$. For this case, the eigenstates of the system are given by
	\begin{equation}
	\displaystyle
	\{\frac{1}{\sqrt{2}}\left(-\ket{00}+\ket{11}\right), \quad \frac{1}{\sqrt{2}}\left(\ket{10}-\ket{01}\right), \quad \frac{1}{\sqrt{2}}\left(\ket{00}+\ket{11}\right),\quad \frac{1}{\sqrt{2}}\left(\ket{01}+\ket{10}\right)\}    
	\end{equation}
	with corresponding eigenenergies  $\{-\kappa,-\kappa,\kappa,\kappa\}$, respectively, having the ground state degenerated.
	
	This pair in a thermal state at temperature T will be CP-local passive $\forall T$. This can be verified numerically using the necessary and sufficient conditions given in our Theorem 1.

	\subsection{ CP-local passivity with coherence in the eigenbasis}

	The main results of previous work on CP-local passivity \cite{frey2014strong} are restricted to states in the form of statistical mixtures of the eigenstates (eigenmixtures). This assumption was motivated from the distinctive role of this type of states in global passivity (a finite state is global passive if and only if it is an eigenmixture $\rho=\sum_k p_k \ket{E_k}\bra{E_k}$ with $p_k\geq p_{k'}$ for $E_k\leq E_{k'}$, in \cite{lenard1978thermodynamical}).
	
	However, here we remove that restriction and show that it is not a necessary condition for CP-local passivity. As a simple example, consider a bipartite system with the Hamiltonian of Eq. (\ref{eq:hamab}), with $\omega=2$ and $\kappa=10$, i.e.
	\begin{equation}
	H_{AB}=\sigma_z^A+\sigma_z^B+\kappa(\sigma_x^A\sigma_x^B+\sigma_y^A\sigma_y^B)/2, 
	\end{equation}
	and the system in a state with coherence in the eigenbasis, for instance
	\begin{equation}
	\rho_{AB}=U\frac{e^{-\beta H_{AB}}}{Z}U^\dagger, 
	\quad
	{\rm with}
	\quad
	U={\rm exp}[i \sigma_x^A\sigma_x^B]
	\end{equation}
	where $Z=\Tr{\left(e^{-\beta H_{AB}}\right)}$. Even though $[\rho_{AB},H_{AB}]\neq 0$, this system is in a CP-local passive state for all $\beta^{-1}<4.95$, this was verified numerically using the result of our Theorem 1.
	
	\section{Tightness of bounds on $p^*$ and $T^*$ (numerical examples)}\label{app:num2}
	
	Here, we show an example that illustrates how far our sufficient bound on $T^*$ from Theorem \ref{th:suff1} may be from the actual threshold temperature $T^*$; and to explore how low our bound $p_b$ for the critical population $p^*$ can be. Let us consider two qubits with Hamiltonian
	\begin{equation}
	H=\sigma_z^A+\sigma_z^B+\kappa\left(\frac{1+\gamma}{2}\sigma_x^A\sigma_x^B+\frac{1-\gamma}{2}\sigma_y^A\sigma_y^B\right).
	\end{equation}
	In particular, let us consider the case of a small coupling anisotropy $\gamma=0.0001$.
	We found numerically the threshold temperature $T^*$ by using the necessary and sufficient conditions presented in our theorem 1, for different values of strength coupling $\kappa$ (see Fig. \ref{fig:bounfT}, in gray). On the other hand, from our inequality Eq. (14) in the main text, obtained from physical characteristics of the system, we can get a lower bound $T_b$ for the critical temperature, see Fig. \ref{fig:bounfT}, which we have verified numerically is in agreement with the results of $T^*$.
	
	\begin{figure}[h]
		\includegraphics[width=0.52\textwidth]{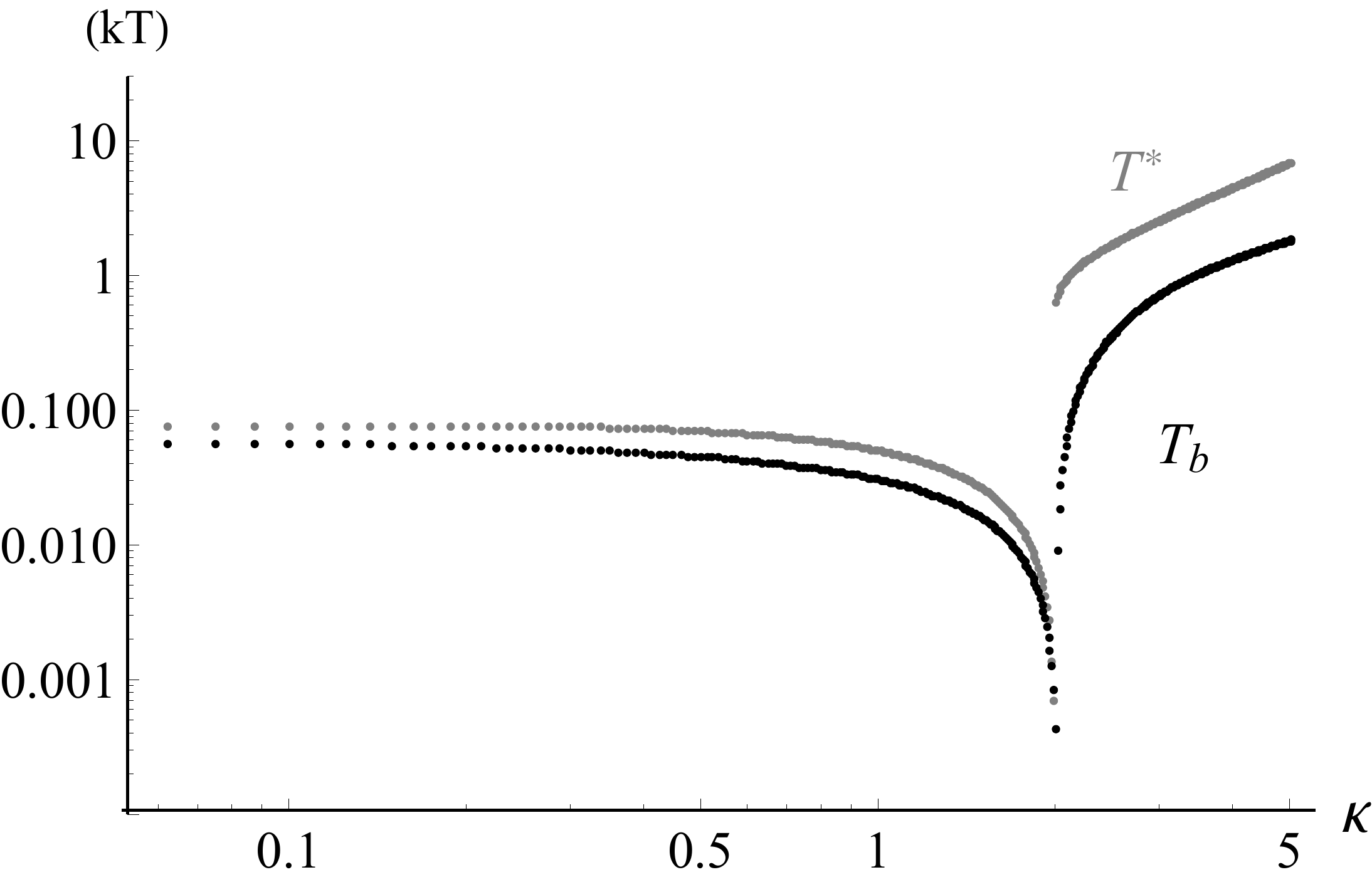}
		\caption{Threshold temperature $T^*$ (in gray), and bound $T_b$ (in black), as a function of the strength coupling $\kappa$, for a pair of qubits with Hamiltonian $H=\sigma_z^A+\sigma_z^B+\kappa\left(\frac{1+\gamma}{2}\sigma_x^A\sigma_x^B+\frac{1-\gamma}{2}\sigma_y^A\sigma_y^B\right)$, in the case of an small anisotropy ($\gamma=0.0001$).}
		\label{fig:bounfT}
	\end{figure}
	


	\section{Proof of Theorem 2   }\label{app:suff}
	
	\begin{customthm}{2}[Threshold ground state population] \label{th:suff1}
		Let the ground state $\ket{E_0}$ of the Hamiltonian $H_{AB}$ be non-degenerate and with full Schmidt rank. All pairs $\{\rho_{AB},H_{AB}\}$ with $\rho_{AB} = \sum_{i} p_i \ket{E_i} \bra{E_i}$ and $p_0 \ge p_*$ are CP-local passive with respect to A, with the threshold ground state population bounded from above by 
		\begin{align}\label{eq:pbound-2}
		p_*  \le  \left( 1 + \frac{E_1   (q^{AB}_{0,\min})^2}{\max_{i\ge 1} \left[ E_i  (q^{AB}_{i,\max})^2 \right]} \right)^{-1} . 
		\end{align}
		$\{  q^{AB}_{i,\alpha} \}_{\alpha=0}^{d_A-1} $ denotes the Schmidt coefficients of $\ket{E_i}$ and $q^{AB}_{i,\min} \equiv \min_{\alpha} \left[ q^{AB}_{i,\alpha} \right]$, $q^{AB}_{i,\max} \equiv \max_{\alpha} \left[ q^{AB}_{i,\alpha} \right]$.
	\end{customthm}
	
	\begin{proof}
		We will show that, as long as the ground-state population $p_0$ of the steady state $\rho_{AB}$ exceeds either of the bounds of Eq. \eqref{eq:pbound-2}, any solution to the optimization problem of Definition 1 in the main text yields a non-negative value for the optimal locally-extractable energy.
		
		
		Let  $\Delta p_i \equiv p'_i - p_i$ denote the change of populations under the action of the local quantum channel, i.e.,  $p'_i \equiv  \braket{E_i | \mathcal E_A \otimes \mathcal I_B \left( \rho \right) | E_i} $. The condition for CP-local passivity of Definition 1  then translates to
		\begin{align}
		\Delta E_{(A)B}^{\mathcal{E}_A}=\sum_i E_i \Delta p_{i} \ge 0  \quad \, \forall \, \, \mathcal{E}_A \,\;. \label{SLP_sufficient_Condition}
		\end{align}
		We define the matrix $S$ with elements
		\begin{align}
		S_{ij} \equiv \braket{E_i | \mathcal E_A \otimes \mathcal I_B \left(\ket{E_j}\bra{E_j} \right) | E_i}
		\end{align}
		such that $p_i' = \sum_j S_{ij} p_j$. Since $\mathcal E_A$ is a quantum channel, the matrix $S$ is stochastic, i.e., $\sum_i S_{ij} = 1$ for all $i$.
		Eq.~\eqref{SLP_sufficient_Condition} can be rewritten in terms of the $S$ matrix as
		\begin{align}
		\sum_{i,j} E_i \left( S_{ij} - \delta_{ij} \right) p_j \ge 0 \,\;. \label{sufficient_S_matrix}
		\end{align}
		A sufficient condition for Eq.~\eqref{sufficient_S_matrix} to hold is
		\begin{align}
		E_1 \left( 1 - S_{00} \right) p_0 \ge \sum_{i \ge 1} E_i \left(1 - S_{ii}  \right) p_i  \,\;. \label{sufficient_S_matrix_2}
		\end{align}
		The above inequality essentially compares the energy difference resulting from two kinds of population transitions: (a) populations leaving the ground state and residing in the first excited state (LHS), and (b) populations leaving all the excited states and residing in the ground state (RHS). Eq. \eqref{sufficient_S_matrix_2} indeed implies CP-local passivity:
		\begin{align*}
		E_1 \left( 1 - S_{00} \right) p_0 - \sum_{i \ge 1} E_i \left(1 - S_{ii}  \right) p_i &\ge 0 \\
		\Longrightarrow  \quad E_1 \sum_{i\ge 1} S_{i0} p_0 -  \sum_{i \ge 0} E_i \left(1 - S_{ii}  \right) p_i &\ge 0 \\
		\Longrightarrow  \quad \sum_{i\ge 1} E_i S_{i0} p_0 +  \sum_{i \ge 0} E_i \left(S_{ii} - 1  \right) p_i &\ge 0 \\
		\Longrightarrow  \quad \sum_{\substack{i,j\ge 0 \\ i\ne j}} E_i S_{ij} p_j +  \sum_{i \ge 0} E_i \left(S_{ii} - 1  \right) p_i &\ge 0 \\
		\Longrightarrow \quad  \eqref{SLP_sufficient_Condition} \quad \Longrightarrow \quad  \text{CP-local passivity}& 
		\end{align*}

		Eq. \eqref{sufficient_S_matrix_2} implicitly depends on the local quantum channel $\mathcal E_A$ through the matrix $S$. We proceed by formulating an $\mathcal E_A$-independent  sufficient condition for the above equation to hold.  The diagonal elements of $S$ can be calculated in terms of the Schmidt decomposition
		\begin{align}
		\ket{E_i} = \sum_{\alpha} \sqrt{q^{AB}_{i,\alpha}} \ket{\alpha _{(i)}}_A \ket{\tilde{\alpha} _{(i)}}_B \,\;.
		\end{align}
		A simple calculation gives
		\begin{align}
		S_{ii} &= \sum_{\alpha,\beta} q^{AB}_{i,\alpha}q^{AB}_{i,\beta} \braket{\alpha_{(i)} |  \mathcal E _A \left( \ket{\alpha_{(i)}} \bra{\beta_{(i)}} \right) |\beta_{(i)} }
		\end{align}
		and hence
		\begin{align}
		1 - S_{ii} &= \sum_{\alpha,\beta} q^{AB}_{i,\alpha}q^{AB}_{i,\beta} \braket{\alpha_{(i)} |  \left(  \mathcal I_A - \mathcal E_A \right) \left( \ket{\alpha_{(i)}} \bra{\beta_{(i)}} \right) |\beta_{(i)} } = \sum_{\alpha,\beta} q^{AB}_{i,\alpha}q^{AB}_{i,\beta} \, r_{i}(\alpha,\beta) \,\;,
		\end{align}
		where
		\begin{align}
		r_{i}(\alpha,\beta) \equiv \braket{\alpha_{(i)} |  \left(  \mathcal I_A - \mathcal E_A \right) \left( \ket{\alpha_{(i)}} \bra{\beta_{(i)}} \right) |\beta_{(i)} } \,\;.
		\end{align}
		Now notice that 
		\begin{align*}
		1 - S_{ii} = \sum_{\alpha,\beta} q^{AB}_{i,\alpha}q^{AB}_{i,\beta} \,\text{Re} \left[ r_{i}(\alpha,\beta) \right] \,\;,
		\end{align*}
		with each $\text{Re} \left[ r_i(\alpha,\beta) \right] \ge 0 $, as demonstrated by the following sequence of inequalities:
		\begin{align*}
		\text{Re} \left[ r_i(\alpha,\beta) \right] &\ge 1 - \left| \braket{\alpha_{(i)}| \mathcal E_A \left( \ket{\alpha_{(i)}} \bra{\beta_{(i)}} \right) | \beta_{(i)}} \right| = 1 - \left| \braket{\,\ket{\alpha_{(i)}}\bra{\beta_{(i)}},\mathcal E_A \left( \ket{\alpha_{(i)}} \bra{\beta_{(i)}} \right)  \, } \right|\\
		&\ge 1 - \left\| \mathcal E_A \left( \ket{\alpha_{(i)}} \bra{\beta_{(i)}} \right) \right\|_2 \ge  1 - \left\| \mathcal E_A  \right\|_{2,1} \ge 1 - \left\| \mathcal E \right\|_{1,1} =0 \,\;,
		\end{align*}
		where the Cauchy-Schwarz inequality was used and the fact that $\mathcal E_A$ is a CPTP map and therefore has a $1-1$ norm $\left\| \mathcal E \right\|_{1,1} = 1$. 
		As a result, we can now bound
		\begin{align}
		(q^{AB}_{i,\min})^2 R_i \le 1 - S_{ii} \le (q^{AB}_{i,\max})^2 R_i \label{S_bounds}
		\end{align}
		where we set
		\begin{align}
		R_i \equiv \sum_{\alpha,\beta} r_i(\alpha,\beta) \,\;.
		\end{align}
		Notice that each $R_i$ is just the (Hilbert-Schmidt) trace, which we define as $ R_i = \Tr \left[ \mathcal I_A - \mathcal E_A  \right] 	$, evaluated with respect to the orthonormal operator  basis $B_i \equiv \left\{ \ket{\alpha_{(i)}}\bra{\beta_{(i)}} \right\}_{\alpha,\beta}$. However, since the value of the trace is independent of the (orthonormal) basis of evaluation $B_i$, we conclude $R \equiv R_i$ does not depend on the index $i$. As a result,
		\begin{align}
		E_1 (q^{AB}_{0,\min})^2 p_0 \ge \sum_{i \ge 1} E_i (q^{AB}_{i,\max})^2 p_i \; \xRightarrow{\eqref{S_bounds}} \; \eqref{sufficient_S_matrix_2} \; \Longrightarrow \text{CP-local passivity} \label{eq:necessary_probabilities} \,\;.
		\end{align}
		Finally, from Eq. \eqref{eq:necessary_probabilities} we can read the desired bound for the threshold ground state population
		\begin{align}
		p_* \le  \left( 1 + \frac{E_1  (q^{AB}_{0,\min})^2}{\max_{i\ge 1} \left[ E_i (q^{AB}_{i,\max})^2 \right]} \right)^{-1} \,\;,
		\end{align}
		since $\sum_{i \ge 1} E_i (q^{AB}_{i,\max})^2 p_i  \le  (1 - p_0) \max_{i\ge 1} \left[ E_i (q^{AB}_{i,\max})^2 \right]$. 
	\end{proof}
	
	\section{Proof of Theorem 3 }
	\label{app:suffcorr}
	
	Before deriving the main result, we need two preliminary Lemmas (Lemmas \ref{le:energy_diff_prop} and \ref{le:norms}). First, Lemma \ref{le:energy_diff_prop} is based on a result of \cite{brandao2016finite} (Lemma \ref{th:locglob}) and it shows that a change in energy of the local Hamiltonian is well approximated by a change of energy of the global Hamiltonian. It relies on the following definition:
	\begin{customdef}{2}[Clustering of correlations]
		We say that the state $\rho$ on a lattice system has $\epsilon(l)$-clustering of correlations if
		\begin{equation}
		\max_{M,N} \left| \tr[M\otimes N \rho]-\tr[M\rho]\tr[N\rho] \right| \le \vert\vert M \vert \vert \, \vert\vert N \vert\vert \, \epsilon(l),
		\end{equation}
		where $M$ has support of region $A$ and $N$ on region $B$, and $l\le \dist (A,B)$.
	\end{customdef}
	
	The relevance of this definition lies in the fact that in many systems of interest $\epsilon(l)$ will be a decaying exponential. 
	If this decay is fast, the following result shows that marginals of thermal states when tracing out a big region can be well-approximated by the marginal of the thermal state of a much smaller lattice.
	
	\begin{lemma}\label{th:locglob} [Theorem 4,  \cite{brandao2016finite}]
		Let H be a local bounded Hamiltonian, $\beta$ an inverse temperature and $\tau_{AB}=e^{-\beta H}/\tr[e^{-\beta H}]$. Let $AB_1B_2$ be a separation of the lattice such that $B_1$ shields $A$ from $B_2$ by a distance of at least $l$. Let $\tau^{AB_1}$ be the Gibbs state on region $AB_1$ only. If the system is $\epsilon(l)$-clustering, then
		\begin{equation}
		\vert\vert \tr_{B}[\tau_{AB}]-\tr_{B_1}[\tau_{AB_1}] \vert \vert_1 \le K |\partial B_2|(\epsilon(l/2)+c_1 e^{-c_2 l}),
		\end{equation}
		where $K>0$ and $c_1,c_2 > 0$ are constant and $|\partial B_2|$ is the size of the boundary between $B_1$ and $B_2$.
	\end{lemma}
	The choice of regions of this lemma is shown for clarity in Fig. 1 in the main text. \noindent This lemma itself relies on the idea of \emph{quantum belief propagation} from \cite{hastings2007quantum}, from which the function $\gamma(l)\equiv c_1 e^{- c_2 l}$ arises. Notice that the boundary will grow polynomially in $l$ for lattice dimension $D>1$.
	
	\noindent We subsequently use this to prove that we can estimate well an energy change in the thermal state of the whole Hamiltonian from the energy change of a thermal state corresponding to smaller part of the system.
	
	\begin{lemma} \label{le:energy_diff_prop}
		Let $AB_1B_2$ be regions in the lattice as defined as in the Lemma  \ref{th:locglob} above, and let $H_{AB}$ be the total Hamiltonian, which we can decompose as
		\begin{equation}
		H_{AB}=H_A+V_{AB_1}+H_{B_1}+V_{B_1B_2}+H_{B_2}
		\end{equation}
		Moreover, let
		$\mathcal{E}_A \otimes \mathcal{I} $ be a CPTP map that acts inside region $S\equiv A \setminus \partial A$ (that is, outside the support of $V_{AB_1}$). Then 
		\begin{equation}
		|\Delta E^{\mathcal{E}_A}_{(A)B} -\Delta^{\mathcal{E}_A} E_{(A)B_1} | \le  \vert\vert H_A \vert\vert\, \vert\vert \mathcal{E}_A-\mathcal{I} \vert\vert_{1,1} \, K |\partial C|(\epsilon(l/2)+\gamma(l/2)),
		\end{equation}	
		where
		\begin{align}
		\Delta E^{\mathcal{E}_A}_{(A)B}&=\tr[H_{AB} (\mathcal{E}_A \otimes \mathcal{I}_{\setminus A} -\mathcal{I}_{AB}) (\tau_{AB})]  \\
		\Delta E^{\mathcal{E}_A}_{(A)B_1}&=\tr[H_{AB_1} (\mathcal{E}_A \otimes \mathcal{I}_{\setminus A} -\mathcal{I}_{AB}) (\tau_{AB_1})]
		\end{align}	
	\end{lemma}		
	\begin{proof}
		Because of where the local map acts, we can write 
		\begin{align}
		\Delta E^{\mathcal{E}_A}_{(A)B}&=\tr[H_{AB} (\mathcal{E}_A \otimes \mathcal{I}_{\setminus A} -\mathcal{I}_{AB}) (\tau_{AB})]=\tr[H_A (\mathcal{E}_A \otimes \mathcal{I}_{\setminus A} -\mathcal{I}_{AB}) (\tr_{B}[\tau_{AB}])]\\
		\Delta E^{\mathcal{E}_A}_{(A)B_1}&=\tr[ H_{AB_1} (\mathcal{E}_A \otimes \mathcal{I}_{\setminus A} -\mathcal{I}_{AB}) (\tau_{AB_1})]=\tr[H_{A} (\mathcal{E}_A \otimes \mathcal{I}_{\setminus A} -\mathcal{I}_{AB}) (\tr_{B_1}[\tau_{AB_1}])].
		\end{align}
		
		Thus we have, using the definition of the norms and Theorem \ref{th:locglob},
		\begin{align}
		|	\Delta E^{\mathcal{E}_A}_{(A)B} -\Delta E^{\mathcal{E}_A}_{(A)B_1} | &= | \tr[H_A(\mathcal{E}_A-\mathcal{I}_A)(\tr_{B}[\tau_{AB}]-\tr_{B_1}[\tau_{AB_1}])|\\
		&\le \vert\vert H_A \vert\vert \, \vert\vert \mathcal{E}_A-\mathcal{I}_A \vert\vert _{1,1}  \, \vert \vert \tr_{B}[\tau_{AB}]-\tr_{B_1}[\tau_{AB_1}] \vert\vert_1 \\
		&\le \vert\vert H_A \vert\vert \,  \vert\vert \mathcal{E}_A-\mathcal{I}_A \vert\vert_{1,1} \, K |\partial B_2|(\epsilon(l/2)+\gamma(l/2)).
		\end{align}
	\end{proof}
	
	We will also use the following technical lemma, which relates different quantities that measure how far a channel is form the identity channel. One is the trace $	\Tr \left(  \mathcal I - \mathcal E  \right)=\sum_i \braket{\alpha_{(i)} |  \left(  \mathcal I - \mathcal E \right) \left( \ket{\alpha_{(i)}} \bra{\beta_{(i)}} \right) |\beta_{(i)} } $, with $\{ \ket{\alpha_{(i)}} \bra{\beta_{(i)}}\}$ a complete basis of the Hilbert space $\mathcal H \cong \mathbb{C}^{d\times d}$, and the other is the $1-1$ norm $\left\| \mathcal I - \mathcal E \right\|_{1,1}$.
	\begin{lemma} \label{le:norms}
		Let $\mathcal E$ be a CPTP map acting on states of  $\mathcal H \cong \mathbb{C}^d$. Then,
		\begin{align}
		\Tr \left(  \mathcal I - \mathcal E  \right)  \ge \frac{1}{d^2} \left\| \mathcal I - \mathcal E \right\|_{1,1}.
		\end{align}
	\end{lemma}
	\begin{proof}
		We can express the $1-1$ norm of a superoperator $\mathcal F \coloneqq \mathcal I - \mathcal E $  as
		\begin{align}
		\left\| \mathcal F \right\|_{1,1} = \sup_{\ket{u},\ket{v}} \left\|  \mathcal F  \left( \ket{u} \bra{v} \right) \, \right\|_1
		\end{align}
		where the supremum is taken over \textit{unit} vectors $\ket{u},\ket{v} \in \mathcal H $ (see, e.g., \cite{watrous2018theory}).
		We can further write 
		\begin{align}
		\left\| \mathcal F \right\|_{1,1} &= \sup_{U}\sup_{\ket{u},\ket{v}} \left| \Tr \left( U  \left[ \mathcal F  \left( \ket{u} \bra{v} \right)  \right] \right) \right| \\
		& = \left| \Tr \left( U_*  \left[ \mathcal F  \left( \ket{u_*} \bra{v_*} \right)  \right] \right) \right|
		\end{align}
		and using the eigenbasis of $U_*$ we can write
		\begin{align}
		U_* &= \sum_i \text{(phase)} \cdot  \ket{i} \bra{i} \\
		\ket{u_*} &= \sum_{j} \text{(phase)} \cdot \sqrt{p^{u}_j} \ket{j} \\
		\ket{v_*} &= \sum_{k} \text{(phase)} \cdot \sqrt{p^{v}_k} \ket{k} \,\;.
		\end{align} 
		Plugging-in the above and setting $\ket{\alpha} \coloneqq d^{-1/2} \sum_i \ket{i}$ we get
		\begin{align}
		\left\| \mathcal F \right\|_{1,1} &\le \sum_{i,j,k} \sqrt{p^{u}_j p^{v}_k} \left| \braket{i| \mathcal F  \left( \ket{j}\bra{k} \right) | i} \right| 
		\le  \sum_{i,j,k,l}  \left| \braket{i| \mathcal F  \left( \ket{j}\bra{k} \right) | l} \right| \\
		& \le d^2  \left| \braket{\alpha| \mathcal F  \left( \ket{\alpha}\bra{\alpha} \right) | \alpha} \right| = d^2 \braket{\alpha| \mathcal F  \left( \ket{\alpha}\bra{\alpha} \right) | \alpha} \le d^2 \Tr \left( \mathcal F \right) \,\;,
		\end{align}
		since the (superoperator) trace can be taken with respect to an orthonormal basis that includes the element $\ket{\alpha}\bra{\alpha}$ and we showed earlier that $\Re \left[ \braket{x| \mathcal F  \left( \ket{x}\bra{y} \right)|y} \right] \ge 0$ (for $\ket{x},\ket{y}$ elements of an orthonormal basis).
		
	\end{proof}
	We are now in a position to prove the central result of the section.
	
	\begin{customthm}{3} \label{th:suffcorr}
		Consider a Hamiltonian $H_{AB}$ as 	
		\begin{equation}\label{eq:hamAB12-2}
		H_{AB}=H_A+V_{AB_1}+H_{B_1}+V_{B_1B_2}+H_{B_2}.
		\end{equation}
		and let $\tau_{AB}^{\beta}=e^{-\beta H_{AB}}/Z_{AB}$ be its thermal state with $\epsilon(l)$-clustering of correlations. There exists a finite temperature $\beta_*$ such that all pairs $\{\tau_{AB}^{\beta},H_{AB}\}$ with $\beta \ge \beta_*$ are CP-local passive with respect to local operations on $S$ if the regions $B_1,B_2$ can be chosen such that
		\begin{align}\label{eq:suffineq-2}
		E_1^{AB_1} \left( q_{0,\min} ^{AB_1}\right)^2 > \lambda(l) \,\;,
		\end{align}
		where 
		\begin{align} 
		\lambda(l) = K d_A^2 \,\vert\vert H_A \vert\vert \, |\partial B_2| \, (\epsilon(l/2)+c_1 e^{-c_2 l}) \,\;. \label{lambda_expr-2}
		\end{align}
		Moreover, $\beta_*$ is such that 
		\begin{align}\label{eq:boundbeta-2}
		\tr[e^{-\beta_* H_{AB_1}}]^{-1} &\le \left(1+\frac{\lambda(l)}{\max_{i\ge 1} \left[ E^{AB_1}_i (q^{AB_1}_{i,\max})^2\right ] } \right) \nonumber \\ & \times  \left( 1 + \frac{E^{AB_1}_1  (q^{AB_1}_{0,\min})^2} {\max_{i\ge 1} \left[ E^{AB_1}_i (q^{AB_1}_{i,\max})^2 \right] }\right)^{-1}.
		\end{align}

		where $K,c_1,c_2 > 0$ are constants.
	\end{customthm}
	\begin{proof}
		
		Let us start by choosing a ground state population
		\begin{equation}\label{eq:p0cond}
		p_0=\left(1+\frac{\lambda(l)}{\max_{i\ge 1} \left[ E^{AB_1}_i (q^{AB_1}_{i,\max})^2\right ] } \right) \nonumber  \left( 1 + \frac{E^{AB_1}_1  (q^{AB_1}_{0,\min})^2} {\max_{i\ge 1} \left[ E^{AB_1}_i (q^{AB_1}_{i,\max})^2 \right] }\right)^{-1}.
		\end{equation}	
		This is equivalent to
		\begin{equation}
		E^{AB_1}_1  (q^{AB_1}_{0,\min})^2 p_0=(1-p_0)\max_{i\ge 1} \left[ E^{AB_1}_i (q^{AB_1}_{i,\max})^2\right ]+\lambda(l)
		\end{equation}
		Since $(1-p_0)\max_{i\ge 1} \left [ E^{AB_1}_i (q^{AB_1}_{i,\max})^2\right ] \ge \sum_{i \ge 1} E^{AB_1}_i (q^{AB_1}_{i,\max})^2 p_i$, it follows that
		\begin{equation}
		E^{AB_1}_1  (q^{AB_1}_{0,\min})^2 p_0 \ge \sum_{i \ge 1} E^{AB_1}_i (q^{AB_1}_{i,\max})^2 p_i + \lambda(l) \label{eq:ineqp2}
		\end{equation}
		Now, notice that from Eq. \eqref{eq:necessary_probabilities} in the proof of Theorem \ref{th:suff1}, the inequality of Eq. \eqref{eq:ineqp2} implies that 
		\begin{align}
		\Delta E^{\mathcal E_A}_{(A)B_1} \ge \lambda(l) \Tr (\mathcal I_A - \mathcal E_A) \,\;,
		\end{align}
		where the change of energy is due to the action of the local channel $\mathcal E_A= \mathcal E_S \otimes \mathcal I _{A/S}$. From Lemma \ref{le:energy_diff_prop} and $\lambda(l)$ as given by Eq~\eqref{lambda_expr-2}, it follows that 
		\begin{align}
		\Delta E^{\mathcal E_A} _{(A)B_1B_2} &\ge \Delta E^{\mathcal E_A} _{(A)B_1} - \frac{\lambda(l)}{d_A^2} \left\|  \mathcal E_A  - \mathcal I _A \right\|_{1,1} \\
		&\ge \lambda(l) \left[ \Tr \left( \mathcal I _A - \mathcal E _A \right) -  \frac{1}{d_A^2} \left\| \mathcal E_A  - \mathcal I _A \right\|_{1,1} \right] \ge 0 \,\;,
		\end{align}
		where for the last inequality we used Lemma \ref{le:norms}.
		
		With this we have shown that a ground state population $p_0$ on the local thermal state obeying Eq. \eqref{eq:p0cond} leads to CP-local passivity on the global thermal state on $AB_1B_2$ with the same temperature. This means that the threshold ground state population of $\tau^\beta_{AB_1}$ that we require is such that $p^* \le p_0$, which corresponds to a temperature $\beta^*$ that obeys Eq. \eqref{eq:boundbeta-2}, finishing the proof.
	\end{proof}	
	
\end{document}